%% file: main-eptcs.tex
\documentclass[submission,copyright,creativecommons]{eptcs}
\providecommand{\event}{MFPS 2021} 
\usepackage{breakurl}             
\usepackage{underscore}           
\usepackage{graphicx}
\usepackage{stmaryrd}
\usepackage{amsmath}
\usepackage{amsthm}
\usepackage{bbm}
\usepackage{amssymb}
\allowdisplaybreaks
\theoremstyle{definition}

\newtheorem{definition}{Definition}
\theoremstyle{plain}
\newtheorem{lemma}{Lemma}
\newtheorem{proposition}{Proposition}
\newtheorem{theorem}{Theorem}
\newcommand{\llb}{\llbracket}
\newcommand{\rrb}{\rrbracket}
\newcommand{\llc}{\{\!\mid}
\newcommand{\rrc}{\mid\!\}}
\newcommand{\ot}{\leftarrow}

\newcommand{\Real}{{\mathbb R}}
\newcommand*{\bind}{\mathrel{\scalebox{0.7}[1]{$>\!\!\!>\!=$}}}

  \title{Monads for Measurable Queries in Probabilistic Databases}
  \author{Swaraj Dash \qquad\qquad Sam Staton
  \institute{University of Oxford\\ Oxford, United Kingdom}
\email{\{swaraj.dash,sam.staton\}@cs.ox.ac.uk}
}
\def\titlerunning{Monads for Measurable Queries in Probabilistic Databases}
\def\authorrunning{S. Dash \& S. Staton}

\begin{document}
\maketitle

\begin{abstract} 
We consider a bag (multiset) monad on the category of standard Borel spaces, and show that it gives a free measurable commutative monoid.
  Firstly, we show that a recent measurability result for probabilistic database queries (Grohe and Lindner, ICDT 2020) follows quickly from the fact that queries can be expressed in monad-based terms. We also extend this measurability result to a fuller query language.
  Secondly, we discuss a distributive law between probability and bag monads, and we illustrate that this is useful for generating probabilistic databases. 
\end{abstract}

\input{intro}
\input{prelims}
\input{bags1}
\input{bags2}

\input{bag-constructions}
\input{datalog}

\bibliographystyle{eptcs}
\bibliography{generic-reconstructed.bib}

\end{document}

%% file: intro.tex
\section{Introduction}
Probabilistic databases cater for uncertainty in data. There may be uncertainty about whether rows should be in a database, or uncertainty about what values certain attributes should have.

For example, consider a database of movies. We might have a table that assigns the gross amount to each movie, which may be quite uncertain for older movies. We might have a table that records which actors appeared in which movies, and there may be uncertainty about whether a particular actor appeared in a given movie. The uncertainty might come from incorrect text processing, for example if the information was scraped off internet forums, or just noise in measurement, e.g.~if the gross amount is difficult to calculate precisely. This is a simple example, but probabilistic databases have applications in other areas of information extraction as well as in scientific data management, medical records, and in data cleaning. See the textbook~\cite{pdb-book} for further examples. 

In this paper, we argue that the semantics of probabilistic databases lies in combining a probability monad,~$P$, with a bag monad $B$ (aka multiset). This builds on the long-established tradition of using monads to structure computational effects in functional programming~\cite{wadler-comprehension,comp-syntax}. 

A good semantic analysis is important in view of the recent work of Grohe and Lindner~\cite{grohe-lindner,grohe-lindner-inf-prob} which builds on~\cite{mcdb,open-world-prob-db}. This new line of work breaks with the traditional approach of having a fixed finite support for the probabilistic database, and argues that the support should be infinite, possibly uncountable. For example, it may be that the gross takings from a movie are approximated as a real number taken from a normal distribution, and it may be that the number of actors appearing in a movie is unknown and unbounded. This leads to semantic complications and introduces issues of measurability. 

\subsection{Two monads}\newcommand{\movieschema}{\mathsf{MovieFact}}\newcommand{\actor}{\mathsf{Actor}}
\newcommand{\cast}{\mathsf{cast}}
\newcommand{\gross}{\mathsf{gross}}
\newcommand{\Movie}{\mathsf{Movie}}
\newcommand{\Actor}{\mathsf{Actor}}
We argue that probabilistic databases are best understood as inhabitants of a set, or space,
\[P(B(X)) \qquad\text{where:}\]
\begin{itemize}\item $X$ is a space of all records (aka rows, tuples) that are allowed according to the schema. For example, in the movie database above, we put $X=\movieschema$ where
\[\movieschema\ =\ \Big(\cast\colon (\mathsf{Actor}\times\mathsf{Movie})\ \uplus\  \gross\colon (\mathsf{Movie}\times \mathbb{R})\Big)\]
  since we can either record that an actor appeared in a movie, or that a movie had a certain gross. (Here we are using a standard notation for tagged disjoint unions.)
\item $B$ is a monad of bags (aka multisets). So $B(X)$ is the space of bags over $X$, and these are the \emph{deterministic} databases for the given schema.
\item $P$ is a probability monad. So $P(B(X))$ is the space of probability distributions (or measures) over the space of deterministic databases, and these are the probabilistic databases for the given schema.
\end{itemize}
In the traditional case, studied in~\cite{pdb-book}, the probability distributions have finite support. In the general setting proposed by~\cite{grohe-lindner-inf-prob}, the support of a distribution is uncountable. This is formalized using measure theory, by placing a $\sigma$-algebra on $X$, by deriving a $\sigma$-algebra on $B(X)$ and $P(B(X))$. We can regard this as moving from the category of sets to a category of measurable spaces. As we will show (Theorem~\ref{thm:bag-monad}), the bag monad $B$ extends to a monad on the category of measurable spaces. We can then regard~$P$ as the Giry monad on the category of measurable spaces~\cite{giry}.

We clarify a subtle point. The support of the distributions in $P(B(X))$ might be infinite, and this means that the set of records that have a chance of appearing in the database can be infinite. But this is a different issue from the sizes of the bags under consideration, which will always be finite. For example, there are infinite possibilities as to what the gross from a movie is, but the number of movies will always be finite. The number of actors in a movie is unbounded, but there is never an infinite cast list for a particular movie. 

\subsection{Measurable queries}
In the deterministic setting, a query (aka view) translates a database from one schema to another. For example, we might ask,
\begin{equation}\label{eqn:query}
  \emph{``Which actors appeared in films that grossed at least \textdollar 200m?''}
\end{equation}
This is a function $q: B(\movieschema)\to B(\actor)$.
For probabilistic databases, the usual approach is to consider queries on deterministic databases, and then lift them to probabilistic databases. Semantically, this can regarded as the functorial action of the monad $P$, which gives a translation between probabilistic databases:
\[P(q) \ :\ P(B(\movieschema))\to P(B(\actor))\]
Notice that if there is uncertainty about whether an actor appeared in a movie, or about what the gross of the movie was, then this will lead to uncertainty about whether that actor should appear in this view.

This functorial action $P(q)$ amounts to pushing forward the probability measure. But this is only legitimate if the query $q$ is measurable. 
In Theorem~\ref{thm:balg}, we show that all queries are measurable provided they are definable in the standard BALG query language for bags~\cite{balg}.

Our proof of measurability is straightforward, because most of the BALG query operations are directly definable from the monad structure of~$B$ (Theorem~\ref{thm:bag-monad}). The remaining operations are easily definable from an fold construction (Theorem~\ref{thm:fold}), which is connected to the fact that~$B(X)$ is the free commutative monoid on~$X$. 

Measurability of a fragment of BALG is perhaps the main technical result of~\cite{grohe-lindner-inf-prob}. That work was groundbreaking, but here we have two additional contributions:
\begin{enumerate}
\item we show that the full language BALG is measurable, which allows us to also treat aggregation queries within the same framework, and
\item we demonstrate that the proof of measurability is almost immediate from the categorical properties of the monad~$B$. 
\end{enumerate}
We give the full details of BALG in Section~\ref{sec:queries}. But for now we note that another way to see that the particular query \eqref{eqn:query} is measurable is that it can be written in the monad comprehension syntax as
\[
  q(b)\ =\ \llc a~|~\cast(a,m) \leftarrow b , \gross(m',r) \leftarrow b,
  m=m',r>200\,000\,000\rrc 
\]
This comprehension syntax works for any strong monad (Section~\ref{sec:compr} and ~\cite{wadler-comprehension}),
indeed it is merely a convenient shorthand for
\[
  b \bind \lambda x.\,b\bind \lambda y. \begin{cases}
    \mathsf{return}(a)&\text{if $x=\cast(a,m)$, $y=\gross(m,r)$ and $r>200\,000\,000$}\\
    \varnothing &\text{otherwise}\end{cases}
\]
where $\bind$ is the monad bind (Kleisli composition) and $\mathsf{return}$ is the monad unit. 
The predicates ($>$, $=$) are well-known to be measurable on the domains where they are used here, and so the query must be measurable. 

As an aside, we remark that much work in the database literature is on computing the results of queries efficiently. In the probabilistic setting, this is even more of a problem. But in this paper (as in~\cite{grohe-lindner-inf-prob}) we are focusing on the semantic aspects. 
\subsection{Generating probabilistic databases}
Having established the measurability of the query language, in Section~\ref{sec:gen} we turn to investigate languages for generating probabilistic databases.
For this we turn to the composite of the monads, $P\circ B$, which we have already shown to be a monad in~\cite{dash} (see also~\cite{jacobs-multisets,keimel-plotkin}). 
As we demonstrate, the language for the monad $P\circ B$ appears to be ideal for generating probabilistic databases, at least as an intermediate language. 

The paradigm for using infinite support probabilistic databases is still under debate, but typically one would begin from a deterministic database, and then add some randomness. Very simple kinds of randomness include
\begin{itemize}
\item adding noise to certain attributes, such as the movie gross, or blood pressure in a medical database;
\item adding or deleting records at random, if there was uncertainty in the accuracy of those records. 
\end{itemize}
We demonstrate how this can be done easily in the monad $P\circ B$. We also investigate a more elaborate model based on a GDatalog program, which translates very cleanly into the language of the $P\circ B$ monad. 

\subsection{Connection with other work on programming semantics}
Our work discusses probabilistic databases in the context of monads and functional programming, and so we bring the general ideas of probabilistic databases to the language of functional probabilistic programming languages.
We have already prototyped our examples simply by implementing a bag monad in Haskell and using a standard Haskell library for probabilistic inference. The idea of applying ideas from probabilistic programming to databases already has some momentum on the practical side, through languages such as BayesDB~\cite{bayesdb} and PClean~\cite{pclean}. Slightly further afield are probabilistic logic languages such as Blog~\cite{russell} and ProbLog~\cite{problog}. 

Probabilistic programming is a general approach for statistics.
Within statistics, inhabitants of $P(B(X))$ are well-known and important, and called `point processes'. 

Further over to the semantic side, we note that the relevance of bags for probability has recently been emphasised by Jacobs~\cite{jacobs-multisets,jacobs-mfps}. Bags are a form of non-determinism, and the problem for combining non-determinism and probability is notoriously subtle, although there has been plenty of recent progress~\cite{goy-petrisan,keimel-plotkin,mio-vignudelli,mow-prob-nondet,non-prob-aut,varacca}. The particular combination we use here is trouble-free. 

\subsection{Summary}
In this paper we show the following. 
\begin{itemize}
\item The Bag monad extends to a strong  monad on standard Borel spaces (Thm.~\ref{thm:bag-monad}).
\item The Bag monad gives a free commutative monoid, and has a `fold' construction (Thm.~\ref{thm:free-mon}, Thm.~\ref{thm:fold}).
\item The BALG language for database queries always yields functions that are measurable (Thm.~\ref{thm:balg}).
\item The composite monad $P\circ B$ combines probability and bags and is useful for generating probabilistic databases. 
\end{itemize}
\paragraph{Acknowledgements.} We are grateful to Peter Lindner for discussions. It has also been helpful to
discuss this work with Martin Grohe, Bart Jacobs, Sean Moss, and Philip Saville.
We acknowledge funding from Royal Society University Research Fellowship, the ERC BLAST grant, and the
Air Force Office of Scientific Research under award number FA9550-21-1-0038. Any opinions, findings, and
conclusions or recommendations expressed in this material are those of the author(s) and do not necessarily
reflect the views of the United States Air Force.

%% file: prelims.tex
\section{Mathematical preliminaries}\label{sec:prelims}
\newcommand*{\dnib}{\mathrel{\scalebox{0.7}[1]{$=\!<\!\!\!<$}}}
\newcommand{\id}{\mathrm{id}}

\subsection{Measure theory}
\begin{definition}
The \emph{Borel sets} form the least collection $\Sigma_\mathbb{R}$ of subsets
of $\mathbb{R}$ containing intervals $(a,b) \subseteq \mathbb{R}$ which is
closed under complementation and countable unions.
\end{definition}

\begin{definition}
A \emph{$\sigma$-algebra} on a set $X$ is a nonempty family $\Sigma_X$ of subsets
of $X$ that is closed under complements and countable unions. 
The pair $(X,\Sigma_X)$ is called a \emph{measurable space} (we just write $X$ when 
$\Sigma_X$ can be inferred from context).

Given $(X,\Sigma_X)$, a \emph{measure} is a function
$\nu: \Sigma_X \rightarrow \mathbb{R}^\infty_+$ such that
for all countable collections of disjoint sets $A_i \in \Sigma_X$,
$\nu\left(\bigcup_i A_i\right) = \sum_i \nu(A_i).$
In particular,  
$\nu(\varnothing) = 0$. It is a \emph{probability measure} if $\nu(X)=1$.
\end{definition}

\begin{definition}
Let $(X,\Sigma_X)$ and $(Y,\Sigma_Y)$ be measurable spaces.
A measurable function $f : X \to Y$ is a function such that
$f^{-1}(U) \in \Sigma_X$ when $U \in \Sigma_Y$.
\end{definition}

\begin{definition}
A measurable space $(X,\Sigma_X)$ is a \emph{standard Borel space} if it is either measurably
isomorphic to $(\mathbb{R},\Sigma_\mathbb{R})$ or it is countable and discrete.
\end{definition}
(This is equivalent to the usual definition of standard Borel spaces, which involves Polish spaces.)


Standard Borel spaces include the measurable spaces of real numbers and the integers, as well as all finite
discrete spaces such as the booleans. So all the measurable spaces that arise in probabilistic databases are
standard Borel, and indeed the restriction to standard Borel spaces is also made in Grohe and Lindner 
(see~\cite[Section~3.1]{grohe-lindner-inf-prob}).
Standard Borel spaces are closed under countable products and countable coproducts. Moreover, the 
equality predicate $X\times X\to \textsf{Bool}$ is measurable when $X$ is a standard Borel space.

\subsection{Monads}
\begin{definition}
  A monad on a category is given by an object
  $TX$ for each object $X$,  a morphism
  $X\to TX$ for each object $X$, and for objects $X$ and $Y$ and
  morphism $f:X\to TY$ a morphism
  $f\dnib:TX\to TY$ is given, satisfying identity and associativity laws (see e.g.~\cite{moggi}). 

  A \emph{strong} monad on a category with products is equipped with a morphism
  $X\times TY \to T(X\times Y)$ that respects the structure (see e.g.~\cite{moggi}).
\end{definition}
The construction $\dnib$ is sometimes called bind or Kleisli composition.
\subsubsection{Monad comprehension notation}\label{sec:compr}
For any strong monad we can use a comprehension notation, which is just syntactic sugar for chaining together
compositions of $\dnib$. The name comes from the fact that this notation resembles set comprehension notation, and when $T$ is the powerset monad on the category of sets, this is exactly set comprehension. But it makes sense for any monad, and is often used with the list monad~\cite{wadler-comprehension}.
As we will see (Section \ref{sec:queries}), for the bag monad, it gives an alternative notation for queries based on products, projection and selection (see also~\cite{comp-syntax}).

Given $f_1:A\to TX_1$, ${f_2:A\times X_1\to T X_2}$, $f_3:A\times X_1\times X_2\to T X_3$,
$f_n:A\times X_1\times\dots\times X_{n-1}\to TX_n$, 
and given $g:A\times X_1\times X_2\times \dots \times X_n\to Y$, we write
\[
  a\ \mapsto \ \llc g(a,x_1,\dots x_n)~|~x_1\leftarrow f_1(a),x_2\leftarrow f_2(a,x_1),\dots,
  x_n\leftarrow f_n(a,x_1,\dots,x_{n-1}) \rrc
\]
for the composite morphism 
\[
 A
 \xrightarrow{\bar{f_1}} 
 T(A\times X_1)
 \xrightarrow {\bar {f_2}\dnib}
 T(A\times X_1\times X_2)
 \to \dots \xrightarrow{\bar {f_n}\dnib}
 T(A\times X_1\times X_2\times \dots \times X_n)
 \xrightarrow {T(g)}
 T(Y)
\]
where
\[\bar{f_i}=A\times X_1\times\dots X_{i-1}
  \xrightarrow{(\id,f_i)} A\times X_1\times \dots X_{i-1}\times TX_i
  \xrightarrow{\text{str}}
  T(A\times X_1\times\dots\times X_i)\text.\]

\subsection{The Giry monad}
The Giry monad~\cite{giry} is a first key monad on measurable spaces. It also restricts to standard Borel spaces.
If $X$ is a measurable space, then $P(X)$ is the set of probability measures on~$X$
equipped with the $\sigma$-algebra generated by
$A_r^U=\{p\in P(X)~|~p(U)\leq r\}$.
The unit is given by the Dirac measures ($\eta(x)(U)=1$ if $x\in U$, otherwise $0$).
The bind is
given by Lebesgue integration: if $f:X\to P(Y)$ then
$(f\dnib p)(U)=\int f(x)(U)\,p(\mathrm{d}x)$.
The strength $s:X\times PY\to P(X\times Y)$ is given by $s(x,p)(U)=p(\{y~|~(x,y)\in U\})$.


%% file: bags1.tex
\section{The bag monoid and monad on measurable spaces}\label{sec:bag}

Let $X$ be a set. A bag, aka multiset, is a finite unordered list of elements of $X$, or more formally an equivalence class of lists under permutation. Equivalently, a bag is a function $b:X\to \mathbb{N}$ such that $\{x~|~b(x)\neq 0\}$ is finite, or more formally it is an integer valued finite measure. 

In this section we will focus on bags in the category of standard Borel spaces. We will show that the bags form a free commutative monoid, and support a `fold' operation. We will also show that the bag construction forms a strong monad. 

We begin by defining the measurable space of bags on some measurable space.

\begin{definition}\label{bag-def}
Let $X$ be a measurable space. Let $BX$ be the set of bags on the set underlying $X$. Equip $BX$ with the least $\sigma$-algebra $\Sigma_{BX}$
containing the generating sets
$A^U_k = \{ b \in BX \mid \text{$b$ contains exactly}$ $\text{$k$ elements in $U$} \}$.
$$\Sigma_{BX} = \sigma(\{ A^U_k \mid U \in \Sigma_X, k \in \mathbb N\})$$
Then $(BX,\Sigma_{BX})$ is the measurable space of bags of $X$.
\end{definition}

Some observations about the space of bags $BX$ are helpful in what follows.
First we note that for any measurable space $X$ we
can decompose the set $BX$ of bags of $X$ into a disjoint union of the set of bags $B_n X$ of size
$n$ for all $n \in \mathbb N$. We can equip each $B_nX$ with the sub-$\sigma$-algebra. Then,
$$BX = \biguplus_{n \in \mathbb N} B_n X.$$
Second we record the following lemma. 
\begin{lemma}\label{lem:fold} Let $X$ be a non-empty standard Borel space. 
  The quotient function $X^n\to B_nX$, which takes a tuple $(x_1,\dots,x_n)$ to the bag
  $\llc x_1\dots x_n\rrc$, is measurable, and has a measurable section $B_nX\to X^n$. 
\end{lemma}
\begin{proof}
  The idea is that we can regard $B_nX$ as a space of sorted lists in $X^n$, as is common in practice in databases. 
  Any standard Borel space is either isomorphic to the reals or countable and discrete.
  All of these spaces have a measurable total order $(<)\subseteq X\times X$.
  There may or may not be a canonical choice for a particular~$X$, but it doesn't matter for the sake of this proof.
  
  We can then use this within the language of measurable functions to write a measurable sorting function $i:X^n\to X^n$ that takes a list and returns the sorted version of it.
  (For example, if $n=2$, let $i(x,y)=(x,y)$ if $x<y$ and otherwise $(y,x)$.)

  As a set-theoretic function, this sorting function $i:X^n\to X^n$ factors through the quotient map,
  \[i\ =\ X^n\xrightarrow{q} B_nX\xrightarrow{s} X^n\text.\]
  It remains to show that these two functions are measurable.
  That $q$ is measurable is well-known, and in fact the $\sigma$-algebra on $B_nX$ can be characterized as $\Sigma_{B_nX}=\{U~|~q^{-1}(U)\in \Sigma_{X^n}\}$~\cite{macchi-bags-measurable,moyal-bagsmeas}. 
  Finally, to see that $s$ is measurable, suppose $U\in \Sigma_{X^n}$, then
  we must show that $s^{-1}(U)\in \Sigma_{B_nX}$, i.e.~that
  $q^{-1}(s^{-1}(U))\in \Sigma_{X^n}$. Since $sq=i$, and $i$ is measurable, we are done. 
\end{proof}

\begin{proposition} The space $(BX, \Sigma_{BX})$ is standard Borel when $X$ is standard Borel.
\end{proposition}
\begin{proof}
If $X$ is standard Borel, then so is $X^n$, since any countable product of standard Borel spaces is standard Borel. So each $B_n(X)$ is standard Borel, since any retract of a standard Borel space is standard Borel. So $B(X)$ is a countable union of standard Borel spaces, hence also standard Borel.
\end{proof}

\noindent
\textbf{Note.} Since all the spaces involved in probabilistic databases are standard Borel,
in the remainder of this paper we only consider standard Borel spaces.

\subsection{Measurable structural recursion on bags}
All the computations we were interested in relied on a form of 
structural recursion over bags, which we now introduce.
This is reminiscent of the fold construction from functional programming.
For example, given a list of integers, it is possible to
compute the sum of its elements by extracting elements starting at the head and calculating a running
sum until we reach the tail.
In this way the function \textsf{sum} can be defined as $\textsf{fold}(\textsf{plus},0)$ where
$\textsf{plus} : \mathbb N \times \mathbb N \to \mathbb N$ plays the role of the accumulating function and 
$0$ is the initial argument provided to \textsf{plus} along with the head of the list. If the list
being considered is empty, the result of the fold is simply the initial argument provided, which
in this case is $0$.
The same approach works for bags too, provided that
our accumulating function will need to be such that the order in which it receives its arguments does
not matter. This leads us to the definition of commutative functions. 
In the rest of this Section we define what it means to \emph{measurably} fold a bag.

\begin{definition}
A function $f : X \times Y \to Y$ is \emph{commutative} if
$$\forall x_1, x_2 \in X, y \in Y.\ f(x_1, f(x_2, y)) = f(x_2, f(x_1, y)).$$
\end{definition}

\begin{definition}
Let $f : X \times Y \to Y$ be a measurable commutative function. Then define
$\mathsf{fold}_f : Y \times BX \to Y$ to be the function which applies the
accumulating function $f$ with initial value $y \in Y$ to each element of $b \in BX$ one-by-one.
Note that the order of selection of elements does not matter as $f$ is commutative.
We first define $\mathsf{fold}^n_f$ for bags of size $n$. When $n=0$, 
$\mathsf{fold}^0_f (y, \varnothing) = y$. For non-zero $n$,
\begin{align*}
&\mathsf{fold}^n_f : Y \times B_n X \to Y\\
&\mathsf{fold}^n_f (y,\llc x_1, \dots, x_n \rrc) = f(x_1, f(x_2, \dots f(x_n, y)\dots)).
\end{align*}
From this, we obtain \textsf{fold} as the unique function coming out of the coproduct of 
each of the $\mathsf{fold}^n_f$'s above, giving us
$$\mathsf{fold}_f : Y \times BX \to Y.$$
\end{definition}
\begin{theorem}\label{thm:fold}
$\mathsf{fold}_f$ is measurable for commutative measurable $f : X \times Y \to Y$.
\end{theorem}
\begin{proof}
  We use Lemma~\ref{lem:fold}.
  First we define fold on lists:
\begin{align*}
&\mathsf{ofold}^n_f : Y \times X^n \to Y\\
&\mathsf{ofold}^n_f (y,( x_1, \dots, x_n )) = f(x_1, f(x_2, \dots f(x_n, y)\dots)).
\end{align*}
This is clearly measurable, because it is just built from composition of measurable functions and product operations. 
Next we note that for any section $B_nX\to X^n$ of the quotient map,
\[\mathsf{fold}^n_f\ =\ Y\times B_nX\xrightarrow{Y\times s}Y\times X^n\xrightarrow{\mathsf{ofold^n_f}} Y\]
The commutativity of $f$ means that the choice of section~$s$ does not matter.
This again is a composition of measurable functions and so $\mathsf{fold}^n_f$ is measurable. The full function $\mathsf{fold}_f:Y\times BX\to Y$ is a copairing of measurable functions, and so it is measurable too. 
\end{proof}
 
\subsection{The space of bags as the free commutative monoid}\label{sec:comm-mon}
In order to define a monoidal structure on the space of bags
we first consider the function \hbox{$\mathsf{add} : X \times BX \to BX$}
which adds a single element to a bag, incrementing its multiplicity by one.
It is clear that \textsf{add} is commutative.
\begin{proposition}\label{prop:add-meas}
$\mathsf{add} : X \times BX \to BX$ is measurable.
\end{proposition}
\begin{proof}
Consider a measurable set $A^U_k \in \Sigma_{BX}$. This is the set of bags with exactly $k$ elements
belonging to $U$. Then $\textsf{add}^{-1}(A^U_k)$ is the set of pairs $(x,b)$ such that 
$\textsf{add}(x,b) \in A^U_k$. In other words, each bag in $A^U_k$ is decomposed into a set
of pairs consisting of an element from the bag and the remaining bag. We consider the cases when
$k=0$ and when $k > 0$. In both these cases the inverse image map is in $\Sigma_{X \times BX}$.
\begin{itemize}
\item $k=0$: Here we consider the set of bags such that no element belongs to $U$. Then it is
guaranteed that any element removed from the bag will be in $\overline U$ and the remaining bag
still in $A^U_0$, resulting in the inverse map being $\overline U \times A^U_0$, which is in
$\Sigma_{X \times BX}$.
\item $k>0$:
Here each bag in the set has a non-zero number of elements in $U$ and as a result we
have two further cases depending on whether or not the element
extracted from the bag is in $U$. If it is not in $U$, the pair consisting
of the element and the remaining bag belongs to $\overline U \times A^U_k$. If it belongs to $U$,
the pair is an element of $U \times A^U_{k-1}$. And so the inverse image of $A^U_k$ under \textsf{add}
is the union of these two sets. Each of these sets is an element of $\Sigma_{X \times BX}$; consequently
so too is their union.
\end{itemize}
\end{proof}

With \textsf{add} as an accumulating function we can define the disjoint union of two bags
as a measurable function by considering the 
fold of \textsf{add} where one bag provides
all the new elements to be \textsf{add}ed to the other bag, which acts as the base case.
\begin{align*}
\uplus& : BX \times BX \to BX\\
\uplus&\ \!(b_1,b_2) = \textsf{fold}_\textsf{add}(b_1,b_2)
\end{align*}

\begin{theorem}\label{thm:free-mon}
For any standard Borel space $X$, $(BX, \uplus, \varnothing)$ is a free commutative monoid.
\end{theorem}
\begin{proof}
First note that $(BX, \uplus, \varnothing)$ is a commutative monoid.
Given any commutative monoid $(Y, +_Y, e_Y)$
and a map $f : X \to Y$ we can define $g = \textsf{fold}_\textsf{monAcc} : Y \times BX \to Y$ where
\textsf{monAcc} is the composite
$$\textsf{monAcc} = X \times Y \xrightarrow{f \times Y} Y \times Y \xrightarrow{+_Y} Y.$$
From this we obtain the unique commutative monoid homomorphism $f^\ast(b) = g(e_Y,b) : BX \to Y$.
\end{proof}
As an aside, we remark that a fold-like operation is sometimes regarded as immediate from the free (commutative) monoid property.
For example, in a cartesian closed category with list objects $X^*$, the space~${Y\to Y}$ is a monoid (under composition), and hence any map $X\to (Y\to Y)$ induces a canonical monoid homomorphism $X^*\to (Y\to Y)$, which is a curried form of fold. 
However, the category of measurable spaces is \emph{not} cartesian closed~\cite{aumann,quasiborel}, and so
we have recorded the existence of fold as a separate fact to the free commutative monoid property.


%% file: bags2.tex
\subsection{The bag monad}

We now use this universal property to describe the structure of the bag monad on standard Borel spaces. 
\begin{itemize}
\item The unit $\eta:X\to B(X)$ is given by
  the singleton bag: $\eta(x)=\llc x\rrc$.
  This is measurable because $\eta^{-1}(A^U_k)$ is $\overline U$ if $k=0$, $U$ if $k=1$, and $\varnothing$
otherwise.
\item The bind is given as follows. Informally, for $f:X\to B(Y)$,
  let $f\dnib:B(X)\to B(Y)$, 
  \[f\dnib \llc x_1\dots x_n\rrc = \biguplus_{i=1}^n f(x_i)\]
  Formally, we apply fold to the composite measurable function 
  \[X\times B(Y) \xrightarrow{f\times B(Y)} B(Y)\times B(Y)\xrightarrow \uplus B(Y)\]
  to get a measurable function $B(X)\times B(Y)\to B(Y)$, and then pass in the empty set as the initial argument. 
Equivalently, the monad multiplication can be given by applying fold to the function
  \[\uplus : B(X)\times B(X)\to B(X)\]
  to get a measurable function $\mu:B(B(X))\to B(X)$, by passing in $\varnothing$ as the initial argument.
\item The strength $X\times B(Y)\to B(X\times Y)$ is given by applying fold to the function
  \[(\mathsf{add},\pi_2):B(X\times Y) \times X\times Y \to B(X\times Y)\times X
  \]
  to get a measurable function $B(X\times Y)\times X\times B(Y)\to B(X\times Y)\times X$, and passing in the empty set as the initial argument and projecting the first result. 
\end{itemize}

As an aside we note that in the statistics literature, it is quite common to regard $B(X)$ as a space of integer valued measures on $X$. With this perspective, regarding $B$ as a monad of measures, the strong monad structure on~$B$ is entirely analogous to the monad structure of the Giry monad~$P$. 
\begin{theorem}\label{thm:bag-monad}
$(B, \eta, \mu)$ is a strong monad on the category of standard Borel spaces.
\end{theorem}

%% file: bag-constructions.tex
\section{Measurable query operations on bags}\label{sec:queries}

\setcounter{footnote}{0}
In the standard theory of database modelling, relations are assumed to be sets, disallowing the existence
of duplicates. Most database software, however, relax this restriction, often to save the cost of
duplicate elimination. BALG (``bag algebra''), an algebra for manipulating bags, was first introduced in~\cite{balg}.
In that paper BALG was presented as an extension of the nested relation algebra (RALG), with a focus on the study
of its expressive power and relative complexity to RALG. The authors showed that BALG as a
query language was more expressive than RALG.

In this Section we will consider the entire BALG query language and show that it extends
to measurable functions on bags. 

For now we briefly review the query language BALG; we discuss these queries and their semantics in more
detail later in this Section. The \textsc{singleton} operation returns a singleton bag consisting of the input.
Restructuring rows of tables is possible using the $\textsc{map}_f$ query, which applies the
function $f$ to every row in the table. The queries \textsc{product},
\textsc{dunion}, \textsc{difference}, \textsc{union}, and \textsc{intersect}
compute the product, disjoint union, difference, union, and intersection of the input
tables respectively.
The \textsc{project} query projects out user-specified columns.
\textsc{flatten} transforms a bag of bags to a bag consisting of the disjoint unions of
all the internal bags. Duplicate elimination,
or deduplication, is possible using the \textsc{dedup} query. Finally, we can compute
the bag of sub-bags of any bag using
\textsc{powerbag}, and the bag of subsets of any bag using \textsc{powerset}.

Given the expressiveness of the BALG query language it comes as no surprise that many operations
can be defined in terms of each other. For example, the powerset of a bag is simply the
deduplicated version of the powerbag.
It is also known that the union and intersection of bags can be defined using the
disjoint union and difference operators.
To this end we will only consider the following
minimal subset of BALG queries, in terms of which all other queries can be defined:
\{{\sc singleton},  {\sc flatten},  {\sc map},  {\sc product},  {\sc project},  {\sc select},  {\sc dunion},
   {\sc difference},  {\sc powerbag},  {\sc dedup}\}.

\textbf{Previous work:} In their work, Grohe and Lindner~\cite{grohe-lindner} considered
BALG$^1$\footnote{It is called BALG$^1$, with superscript 1.}, a subset of BALG restricted to bags of
nesting level 1. That is, the queries of BALG$^1$ are defined on bags of type $BX$ where $X$ cannot
have another type $BY$ in its definition. The minimal set of queries for BALG$^1$ is the same set we
consider here minus \textsc{flatten} and \textsc{powerbag} since they operate on bags of bags.
In their work Grohe and Lindner showed that BALG$^1$ queries extend to measurable functions on bags.
We generalise their results and show, using our monadic and monoidal structure on bags, that all of
BALG extends to measurable functions on bags, and give a clearer picture of how it comes together. Furthermore, we
discuss the actions of grouping and aggregation as measurable queries in BALG.

\subsection{Measurability of BALG queries}

We provide a semantics to BALG queries by mapping each query to a measurable
function on bags. The measurability of the semantics of the 
{\sc singleton},  {\sc flatten},  {\sc map},  {\sc product},  {\sc project}, and {\sc select} queries
is guaranteed by defining their semantics as monad comprehensions.
The measurability of the semantics of the remaining 
queries, {\sc dunion}, {\sc difference},  {\sc powerbag}, and {\sc dedup},
is obtained by defining their semantics using our fold construction introduced in Section~\ref{sec:bag}.
Note that commutativity holds for all the measurable accumulating functions in the
\textsf{fold}-based definitions to follow.
The condition of commutativity is easy to check. 

\paragraph{Bagging and flattening} 
The semantics for the \textsc{singleton} and \textsc{flatten} queries are given by the
unit $\eta^B$ and multiplication $\mu^B$ maps for the bag monad $B$.
The measurability of these maps is proved in Theorem \ref{thm:bag-monad}. 
\begin{align*}
&\llb \textsc{singleton} \rrb : X \to BX \qquad\qquad\qquad\qquad
\llb \textsc{flatten} \rrb : B^2X \to BX \\
&\llb \textsc{singleton} \rrb (x) = \eta^B_X(x) = \llc x \rrc\qquad\qquad
\llb \textsc{flatten} \rrb (b) = \mu^B_X(b)
\ \\
\intertext{
\paragraph{Restructuring}
The \textsc{map}$_f$ query comes parametrized with a measurable function $f : X \to Y$
and its semantics is simply the functorial action of $B$ on $f$, which yields
yet another measurable map.
This can equivalently be written using monad comprehension notation.
}
&\llb \textsc{map}_f \rrb : BX \to BY \qquad\llb \textsc{map}_f \rrb (b) = B(f)(b)
\intertext{
This can be neatly written using monad comprehension notation:
$\llb \textsc{map}_f \rrb (b) = \llc f(x) \mid x \ot b \rrc$.}
\intertext{
\vspace{-2em}
\paragraph{Product and projection}
Monad comprehensions make it straightforward to define the \textsc{product}
of two bags where the arity of the resultant schema is the sum of the arities of the input schemas.
}
&\llb \textsc{product} \rrb : B(X_1 \times \dots \times X_m) \times B(Y_1 \times \dots \times Y_n)
\to B(X_1 \times \dots \times X_m \times Y_1 \times \dots \times Y_n)\\
&\llb \textsc{product} \rrb (b_1, b_2) =
\llc (x_1,\dots,x_m,y_1,\dots,y_n) \mid (x_1,\dots,x_m) \ot b_1,\ (y_1,\dots,y_n) \ot b_2 \rrc
\intertext{
A similar treatment can be given to the \textsc{project}$_{i_1, \dots, i_k}$ query
which projects out the $i_1,\dots,i_k$ indices of the input schema.
}
&\llb \textsc{project}_{i_1, \dots, i_k} \rrb : B(X_1 \times \dots \times X_n) \to
                                                B(X_{i_1} \times \dots \times X_{i_k}) \\
&\llb \textsc{project}_{i_1, \dots, i_k} \rrb (b) =
 \llc (x_{i_1},\dots,x_{i_k}) \mid (x_1,\dots,x_n) \ot b \rrc
\intertext{
\paragraph{Selection}
$\textsc{select}_\psi$ is parametrized by a measurable Boolean predicate
$\psi : X \to \textsf{Bool}$ and filters out the rows in the input table satisfying $\psi$. 
We can lift $\psi$ to the measurable function $\hat\psi : X \to B1$, where $1 = \{ \star \}$ is
the singleton space with a unique element. $\hat\psi$ evaluates to $\llc \star \rrc$ (resp. $\varnothing$)
when $\psi$ evaluates to \textsf{True} (resp. \textsf{False}). This construction enables us
to define the semantics of $\textsc{select}_\psi$ as a monad comprehension where rows
not satisfying $\psi$ do not get included due to $\hat\psi$ evaluating to the empty bag.
(Define $\textsf{filter}_\psi$ to be this map.)
}
&\llb \textsc{select}_\psi \rrb : BX \to BX \qquad \llb \textsc{select}_\psi \rrb (b) = \textsf{filter}_\psi (b) = \llc x \mid x \ot b,\ \star \ot \hat\psi(x) \rrc
\intertext{
In monad comprehension syntax, for a monad with a given zero element (e.g.~$\varnothing\in BX$), a shorthand notation                  $ \llc x \mid x \ot b,\ \psi(x) \rrc                           $    is often used. 
\paragraph{Disjoint union}
\textsc{dunion} simply computes the disjoint union of its arguments.
In Section~\ref{sec:comm-mon} we defined the measurable disjoint union $\uplus : BX \times BX \to BX$ as
$\textsf{fold}_\textsf{add}$.
}
&\llb \textsc{dunion} \rrb : BX \times BX \to BX \qquad \llb \textsc{dunion} \rrb (b_1, b_2) = b_1 \uplus b_2 = \textsf{fold}_\textsf{add}(b_1, b_2)
\intertext{
\paragraph{Bag difference}
Bag difference is the first operation we consider that is not quite immediate from the monad and monoid structure. The idea is that, for instance,
$\llb \textsc{difference} \rrb(\llc1,1,2\rrc, \llc1,2,3\rrc)=\llc 1\rrc$.
Defining bag difference as a fold requires a little care. To this end, we first
define the measurable function \textsf{remove} which takes an element and a bag as input and returns either
the same bag if that element did not belong to the bag, or a modified bag with one fewer instance of the
given element. 
\textsf{remove} is defined as a fold over the accumulating function \textsf{remAcc}, which maintains
a triple as an accumulator:
\begin{itemize}
\item the value $x_\textsf{rem} \in X$ to be removed,
\item a Boolean value indicating whether or not $x_\textsf{rem}$ has already been removed (in
order to prevent us from removing $x_\textsf{rem}$ more than once),
\item and the resulting bag (which is initially empty) to which elements are added.
\end{itemize}
 Each of the three cases has been defined by combining tupling and \textsf{add}, both of which are measurable
functions. From this it follows that \textsf{remAcc} is measurable.
}
&\textsf{remAcc} : X \times ((X \times \textsf{Bool}) \times BX) \to (X \times \textsf{Bool}) \times BX \\
&\textsf{remAcc} (x, ((x_\textsf{rem}, c), b)) = 
\begin{cases}
(( x_\textsf{rem},\textsf{True}  ), \textsf{add}(x,b)) & \text{if } c = \textsf{True} \\
(( x_\textsf{rem},\textsf{True}  ), b) & \text{if } x = x_\textsf{rem} \\
(( x_\textsf{rem},\textsf{False} ), \textsf{add}(x,b)) & \text{otherwise}
\end{cases}
\intertext{
To get the final bag after removal we project out the second element of the pair returned by \textsf{remAcc}.
}
&\textsf{remove} : X \times BX \to BX \qquad
\textsf{remove} (x_\textsf{rem}, b) =
\pi_2 \left(\textsf{fold}_\textsf{remAcc}(((x_\textsf{rem}, \textsf{False}),\varnothing), b)\right)
\intertext{
Using \textsf{remove} we define the bag difference of $b_1$ and $b_2$ by letting $b_1$ be the 
initial input
and from it \textsf{remove}-ing each element in $b_2$ one-by-one. The measurability of bag difference
follows from the commutativity and measurability of \textsf{remove}.
}
&\llb \textsc{difference} \rrb : BX \times BX \to BX \qquad\llb \textsc{difference} \rrb (b_1, b_2) = \textsf{fold}_\textsf{remove}(b_1, b_2)
\intertext{
\paragraph{Powerbag}
For example, the powerbag of the bag $\llc 1, 1 \rrc$ is
$\llc \varnothing, \llc 1 \rrc, \llc 1 \rrc, \llc 1,1 \rrc \rrc$.
                                                          The \textsc{powerbag} of a bag is defined by folding over the accumulating function
\textsf{powerAcc} where for every new element $x$ added to the accumulating bag of bags $b_0$
we first \textsf{add} $x$ to every bag in $b_0$ and then take the disjoint union with the
initial $b_0$.
We define $\textsf{add}_x$ to be the $x$-section of \textsf{add} (that is,
$\textsf{add}_x(b) = \textsf{add}(x,b)$). Sections of measurable functions are measurable.
}
&\begin{array}{@{}l@{\qquad}l}\textsf{powerAcc} : X \times B^2X \to B^2X &\llb \textsc{powerbag} \rrb : BX \to B^2X \\
  \textsf{powerAcc} (x,b_0) = b_0 \uplus B(\textsf{add}_x)(b_0)
  &\llb \textsc{powerbag} \rrb (b) = \textsf{fold}_\textsf{powerAcc}(\llc \varnothing \rrc, b)\end{array}
\intertext{
\paragraph{Deduplication}
In order to \textsc{dedup}licate a bag we recurse over its elements and add them to 
an accumulating bag one-by-one. We avoid multiplicities greater than one in the final bag
by first filtering out the value we want to add from the accumulating bag before adding it,
which is reflected in the definition of \textsf{dedupAcc}.
}
  &\begin{array}{@{}l@{\qquad}l}\textsf{dedupAcc} : X \times BX \to BX
&\llb \textsc{dedup} \rrb : BX \to BX \\     
     \textsf{dedupAcc} (x,b) = \textsf{add}(x, \textsf{filter}_{\not= x}(b))
     &\llb \textsc{dedup} \rrb (b) = \textsf{fold}_\textsf{dedupAcc}(\varnothing, b)
       \end{array}
\end{align*}
The function $\not=\!\!x : X \to \textsf{Bool}$ is measurable since
$\not=\!\!x^{-1}(\{\textsf{True}\}) = X \setminus \{x\}$ 
and
$\not=\!\!x^{-1}(\{\textsf{False}\}) = \{x\}$, both of which are measurable sets (due to
$X$ being standard Borel).
\begin{theorem}\label{thm:balg}
BALG queries yield measurable functions on bags.
\end{theorem}

\subsection{Grouping and aggregation}
Consider, for example, the table $\textsf{cast} : \textsf{Actor} \times \textsf{Movie}$ from the 
database \textsf{MovieFact} introduced at the start of this paper. A natural query that one
may want to compute is, \emph{``How many movies has each actor appeared in?''} In order to calculate
the answer to this we first need to be able to group actors with the bag of all the movies they
appeared in. To this resultant table we can map a \textsf{size} function to the second column to
get the numbers we need. Here we introduce a \textsc{group} query to BALG and show that it
is a measurable operation on bags.

\begin{definition}
The $\textsc{group}_{p_1,p_2}$ query acts on tables of schema $X_1 \times \dots \times X_k$
and is parametrized by two projection functions
$p_1 : X_1 \times \dots \times X_k \to X_{i_1} \times \dots \times X_{i_m}$
and
$p_2 : X_1 \times \dots \times X_k \to X_{j_1} \times \dots \times X_{j_n}$.
The result of this query is a table with schema
$ (X_{i_1} \times \dots X_{i_m}) \times B(X_{j_1} \times \dots \times X_{j_n})$
where the elements of the first column are paired with the bag of elements they
were related to in the input table. In other words, we group the rows of the 
table by the elements in the $p_1$-projection of the table.

The measurable bag-semantics for $\textsc{group}_{p_1,p_2}$ is given by
\begin{align*}
&\llb \textsc{group}_{p_1,p_2} \rrb : B(X_1 \times \dots \times X_k) \to 
 B((X_{i_1} \times \dots X_{i_m}) \times B(X_{j_1} \times \dots \times X_{j_n})) \\
&\llb \textsc{group}_{p_1,p_2} \rrb (b) =
\llc (i, B(p_2)(\textsf{filter}_{\lambda x. p_1(x) \equiv i}(b))) \mid
                                                                                      i \ot \llb \textsc{dedup} \rrb(B(p_1)(b)) \rrc
  \\
  &\phantom{\llb \textsc{group}_{p_1,p_2} \rrb (b) } = 
    \llc (i, \llc p_2(x)~|~x\leftarrow b,p_1(x)=i\rrc)~ |~ i \leftarrow \llb\textsc{dedup}\rrb(\llc p_1(x)~|~x \leftarrow b\rrc )\rrc
\end{align*}
\end{definition}
In the monad comprehension we first project out the columns of interest using $p_1$ and deduplicate
the resultant bag. From this bag we extract out the rows indices by which we index the rows of 
the input bag.
For each index $i$ we return the pair consisting of $i$ along with the $p_2$-projection of the input
where the $p_1$-projection of the rows is equal to $i$. We can conclude that this query is measurable
by defining it as a monad comprehension composed of other measurable functions.

Recall the actor grouping example suggested earlier.
Given an input bag from $B(\textsf{Actor} \times \textsf{Movie})$
we can apply $\textsc{group}_{\pi_1,\pi_2}$ to create a table of rows relating actors to
the bag of movies they appeared in. To this table we can apply
$\textsc{map}_{\lambda (x,y). (x, \textsf{size}(y))}$ to arrive at the final result. 
The function $\textsf{size} : B\ \! \textsf{Movie} \to \mathbb N$ can be measurably defined as a fold,
for example.

A second option for defining an group/aggregation query comes from an extension to monad comprehensions 
in Haskell where the syntax has been extended with the keywords {\tt group by}~\cite{wadler-spj}. This extension works for any strong monad, but the user needs to provide a grouping function which, in our case, needs to map
a bag on $X$ to a bag of bags on $X$ where each sub-bag contains the same element. This can be written in
BALG as
\begin{align*}\textsc{group'}(b) &= \textsc{map}_{\lambda i. \textsc{select}_{\lambda x. x \equiv i}(b)}(\textsc{dedup}(b))
\end{align*}
So $\llb \textsc{group'} \rrb : BX \to B^2 X$, with:
\begin{align*} \llb \textsc{group'}(b) \rrb &= \llc \llc x ~|~x\leftarrow b,x=i\rrc ~|~i\leftarrow \textsc{dedup}(b)\rrc
\end{align*}
The entire actor query can now be concisely written in the notation of~\cite{wadler-spj} as
$$\llc (\textsf{the}\ a, \textsf{size}\ m) \mid (a,m)
   \ot \textsf{actorMovieTable},\ \texttt{group by}\ a \rrc.$$
This modified comprehension syntax of~\cite{wadler-spj} works by implicitly changing the types of $a$ and $m$ 
from \textsf{Actor} and \textsf{Movie} in the right half of the comprehension
to $a\in B\ \!\textsf{Actor}$ and $m\in B\ \!\textsf{Movie}$ on the left half.
This allows us to apply the measurable aggregation function $\mathsf{size}:BX\to \mathbb{N}$ to~$m$.
The aggregation function used on $a$ is $\mathsf{the} : B\ \!\textsf{Actor} \to \textsf{Actor}$, which is some measurable function such that $\mathsf{the}(b)=x$ when $b$ is a bag that only contains copies of~$x$. (For example, we could sort $b$ and return the first element.)


%% file: datalog.tex
\section{Generating probabilistic databases}\label{sec:gen}
The main focus of this paper has been measurable query languages (Theorem~\ref{thm:balg}).
We now turn to the question of where probabilistic databases come from in the first place, particularly in the setting where they have infinite support.
A good language for generating infinite probabilistic databases remains a topic of active research, but we now illustrate that the monads for probability~$P$ and bags~$B$ could be the basis of a good intermediate language.

\subsection{A distributive law}
\newcommand{\distrl}{\mathsf{distr}}
\newcommand{\normal}{\mathit{normal}}
\newcommand{\bernoulli}{\mathit{bernoulli}}
\newcommand{\poisson}{\mathit{poisson}}
\newcommand{\distrhelp}{\mathsf{distrAcc}}
We recall the distributive law $\distrl$ between the monads $P$ and $B$ that we provided in an earlier paper~\cite{dash} (see also \cite{jacobs-multisets,keimel-plotkin,zwart}):
\[
  \distrl : B( P(X)) \to P (B(X))
\]
Using our fold technology (Thm.~\ref{thm:fold}) we can define the distributive law
as $\distrl=\mathsf{fold}_{\distrhelp}$,  where
\begin{align*}
  &  \distrhelp : P(X)\times P(B(X)) \to P(B(X))\\
  &\distrhelp \ =\ P(X)\times P(B(X)) \xrightarrow{\text{double strength}} P(X\times B(X))
    \xrightarrow {P(\mathsf{add})} P(B(X))
\end{align*}
As usual, this distributive law determines a monad structure on $P\circ B$ \cite{distr}. 

\subsection{Randomizing attributes}
For a first example of a probabilistic database, suppose we are given a deterministic database of (movie,gross) pairs.
We may then decide that the gross figure is inaccurate and should be subject to a noise from a normal distribution, yielding a probabilistic database. This can be done categorically by the following map:
\[
  B(\Movie\times \Real)\xrightarrow{B(\Movie\times \normal)}
  B(\Movie\times P(\Real))\xrightarrow {\text{strength}}
  B(P(\Movie\times \Real)) \xrightarrow {\text{distr}}
  P(B(\Movie\times \Real))
\]
Since $PB$ is a monad, we can use comprehension notation for it, and equivalently write the above generation method
 $\mathit{addNoise}: B(\Movie\times \Real)\to   P(B(\Movie\times \Real))$
as
\[
  \mathit{addNoise}(b)=\llc (m,r')~|~(m,r)\leftarrow b,r'\leftarrow \normal(r,100000)\rrc 
\]
Here we are implictly casting $b\in B(X)$ to $b\in P(B(X))$ and $\normal(r,100000)\in P(\Real)$ to $P(B(\Real))$, implicitly using the units
$\eta_P B:B(X)\to P(B(X))$ and $P\eta_B : P(X)\to P(B(X))$. 

Another use-case for random attributes, studied in~\cite{grohe-lindner-inf-prob}, is to deal with null attributes by drawing them randomly from a vague prior distribution. This would also be easy to express using the $PB$ monad. 
\subsection{Adding random records}
\newcommand{\ractor}{\mathit{r\text-actor}}
\newcommand{\rmovie}{\mathit{r\text-movie}}
We can also add and remove random records straightforwardly.
We note a few helpful facts.
\begin{itemize}
  \item The disjoint union operation $\uplus:B(X)\times B(X)\to B(X)$ lifts to
    $\uplus :P(B(X))\times P(B(X))\to P(B(X))$ by composing with the strength of $P$.
    In this way the composite monad $PB$ has a commutative monoid structure.
  \item 
    Since $B(1)\cong \mathbb{N}$, we can regard the Poisson distribution as a map
    in $\Real\to P(B(1))$, parameterized by rate. 
  \end{itemize}
  Now supposing we also have reasonably uniform distributions
    $\ractor:P(\Actor)$ and $\rmovie:P(\Movie)$, 
we can delete some credits and generate random additional actors for movies, modelling the fact that some actors are unlisted:
\begin{align*}
  \mathit{addRemove}(b)=\quad
  & \llc (a,m)~|~(a,m)\leftarrow b, 1\leftarrow \mathit{bernoulli}(0.9)\rrc 
    \\\ \uplus\  &\llc (a,m)~|~n\leftarrow \mathit{poisson}(3),a\leftarrow \ractor, m\leftarrow \rmovie\rrc 
\end{align*}
The first line deletes random rows with probability $0.1$, and the second line
adds in some extra actors (on average, 3 extra actors).
Of course, a more sophisticated model could take into account other prior information such as relationships and ages between actors. 

\subsection{Towards GDatalog}
\newcommand{\earthquake}{\mathsf{earthquake}}
\newcommand{\burglary}{\mathsf{burglary}}
\newcommand{\addres}{\mathsf{address}}
\newcommand{\crimechance}{\mathsf{crimechance}}
\newcommand{\trigger}{\mathsf{trigger}}
\newcommand{\alarm}{\mathsf{alarm}}
\newcommand{\City}{\mathsf{City}}
\newcommand{\House}{\mathsf{House}}
The GDatalog language has recently been proposed as a generative language for probabilistic databases~\cite{barany-gdatalog,grohe-gdatalog}. The language combines datalog-style features with continuous probability distributions.

In general, GDatalog is recursive. We have not treated recursion in this paper, so we focus on the non-recursive fragment. This can easily be translated into the $PB$ monad.
For example, consider the following GDatalog program taken from \cite{grohe-gdatalog}. The idea is to simulate possibly faulty burglar alarms which either go off because of a burglary or because of an earthquake.
\begin{align*}
  &\earthquake(c,\bernoulli(0.1))\leftarrow \crimechance(c,r)
\\
  &
    \burglary(x,\bernoulli(r)) \leftarrow \addres(x,c),\crimechance(c,r)
  \\
  &
    \trigger(x,\bernoulli(0.6)) \leftarrow \addres(x,c),\earthquake(c,1)
  \\
  &  \trigger(x,\bernoulli(0.9)) \leftarrow \burglary(x,1)
  \\
  &\alarm(x) \leftarrow \trigger(x,1)
  \end{align*}
We will regard this as transforming a deterministic database into a probabilistic one, $B(X) \to P(B(X))$ where
\[\begin{array}{rll}
  X\ = & &(\earthquake:\City\times 2)\uplus
(\crimechance:\City \times [0,1])\uplus
(\addres:\House\times \City)\\&\uplus&
(\burglary: \House\times 2)\uplus
(\trigger: \House\times 2)\uplus 
                                       (\alarm : \House)\end{array}
                                     \]
The GDatalog program can be translated almost verbatim as a sequence of definitions of a
probabilistic database in $PB(X)$ using monad comprehensions, starting from $b\in B(X)$, as follows.
  \[\begin{array}{rl@{}l}
    b_1 &= b \ &\uplus\ \llc \earthquake(c,z)~|~\crimechance(c,r)\leftarrow b, z\leftarrow \bernoulli(0.1)\rrc 
    \\  b_2 &= b_1\ &\uplus\  
\llc  \burglary(x,z)~|~\crimechance(c,r)\leftarrow b_1, \addres(x,c')\leftarrow b_1,c=c',z\leftarrow \bernoulli(r)\rrc 
    \\b_3 &= b_2 &\uplus\ 
\llc  \trigger (x,z)~|~\addres(x,c)\leftarrow b_2, \earthquake(c',1)\leftarrow b_2,c=c',z\leftarrow \bernoulli(0.6)\rrc 
    \\b_4 &= b_3 &\uplus\ 
\llc  \trigger (x,z)~|~\burglary(x,1)\leftarrow b_3, z\leftarrow \bernoulli(0.9)\rrc 
    \\b_5 &= b_4 &\uplus\ 
\llc  \alarm (x)~|~\trigger(x,1)\leftarrow b_4\rrc 
  \end{array}\]
One of the main results of~\cite{grohe-gdatalog} is that GDatalog programs yield proper probabilistic databases, that is, that all the constructions are measurable. For examples such as this, in the non-recursive fragment, this measurability is immediate from the fact that we are programming in the $PB$ monad, where everything is measurable.

\section{Summary and outlook}

We have shown that the bag monad on standard Borel spaces is strong (Thm.~\ref{thm:bag-monad}) and supports a fold operation (Thm.~\ref{thm:fold}) which is connected to its characterization as the free commutative monoid. We have used this to show straightforwardly that all the bag query operations of BALG yield measurable queries~(Thm.~\ref{thm:balg}), and so they are all safe to use in defining queries of probabilistic databases. This affirms the results in~\cite{grohe-lindner}, generalizing them to the full BALG language and clarifying the measurability arguments by factoring them through the measurability of the bag monad. Finally, in Section~\ref{sec:gen} we have argued by illustrations that the combination of the bag and probability monads gives a powerful intermediate language for processes that generate probabilistic databases.
